\documentclass[a4paper]{article}
\usepackage{amsmath}
\usepackage{parskip}
\usepackage{authblk}
\usepackage[a4paper, total={5.45in, 8.4in}]{geometry}
\usepackage{amssymb}
\usepackage{amsfonts}
\usepackage{amsthm}
\usepackage{xcolor}
\usepackage[hidelinks]{hyperref}
\usepackage{cite}
\usepackage{graphicx}
\usepackage{subcaption}
\usepackage{caption}
\usepackage{tikz-cd}
\usepackage{enumerate}
\usepackage[shortlabels]{enumitem}
\usepackage{mathabx}

\graphicspath{ {./figures/} }
\DeclareMathOperator*{\argmax}{argmax}

\stepcounter{footnote}

\title{A Symmetry-based Framework for Model Selection of Coral Reef Population Growth Models}
\author{Reemon Spector\thanks{\textbf{Affiliation:} Mathematical Institute, University of Oxford. \textbf{Email:} \href{mailto:spector@maths.ox.ac.uk}{\texttt{spector@maths.ox.ac.uk}}}}
\date{}

\theoremstyle{definition}
    \newtheorem{prop}{Proposition}[section]
    \newtheorem*{abst}{Abstract}
    \newtheorem*{ack}{Acknowledgements}
    \newtheorem{defn}[prop]{Definition}
    \newtheorem{nota}[prop]{Notation}
    \newtheorem{model}[prop]{Model}
    \newtheorem{exam}[prop]{Example}
    \newtheorem{alg}[prop]{Algorithm}
    \newtheorem{ques}[prop]{Question}

\theoremstyle{plain}

    \newtheorem{thm}[prop]{Theorem}

\theoremstyle{remark}
    \newtheorem{rem}[prop]{Remark}

\begin{document}

\maketitle

\begin{abst}
    The problem of selecting a model given a set of candidates remains a challenging one that pervades many scientific fields. We employ techniques from the theory of Lie groups to analyse the symmetries in differential equation models of population growth, with the aim of informing the model selection problem. To illustrate the use of Lie symmetries in model selection, we apply them to simulated data and to coral reef data from the Great Barrier Reef, demonstrating that the trivial symmetries can distinguish between candidate models. A method for finding locally optimal parameters for multi-parameter symmetries is presented, and the paper concludes with related results, some open problems, and avenues of further research.

\end{abst}

\section{Introduction}
    \paragraph{}
    Population growth models are ubiquitous in the scientific literature, particularly so in mathematical biology. The problem of selecting a model given a set of candidates is challenging \cite{Ger} and has been tackled using tools from a multitude of fields, including classical methods such as regression and information-theoretic criteria (which quantify the trade-off between model complexity and quality of fit) such as the Akaike Information Criterion (AIC) \cite{Aka,Boz,Bur} or the Bayesian Information Criterion (BIC) \cite{Sch}. More recently, the model selection problem has been considered through the lens of structural \cite{CdB, MV, Mer, Yat} and practical parameter identifiability \cite{Sim}.

    Even in the setting where several models provide an equally good fit (by classical measures) to data, selecting a particular model from a set of candidates can have significant quantitative consequences for both explanation (through underlying mechanistic assumptions \cite{Lai} or through estimated parameters) and prediction \cite{Sim}.

    In this paper we implement a framework for model selection, as described in \cite{Bor} and \cite{OBC}, using techniques from the theory of Lie groups which allow us to exploit
    the fact that the solution space to an ODE model is closed under the actions of its symmetries, but not necessarily under other transformations. In particular, we show that the so-called \textit{trivial symmetry}, a transformation inherent to many ODE models, can be used to carry out model selection in this framework, removing the potentially significant bottleneck of finding symmetries for a model. We also introduce the \textit{disagreement coefficient} as a tool used to determine (locally) optimal parameters for this framework for model selection.
    
    To do this, we provide the necessary technical background, which is readily found in the literature \cite{Hyd,Olv}.

\section{Background}
    \begin{defn}
        \label{symmetry definition}
        Given an ODE $\dot{y} = \omega(t,y)$, we define the \textit{total space} $E = T \times Y$ to be the ambient space of all possible values of the independent variable $t \in T$, and the dependent variable $y \in Y$. This is extended to the \textit{first jet space} $\mathcal{J}^{(1)} = T \times Y \times Y^{(1)}$, that also includes the derivative $\dot{y} \in Y^{(1)}$. The \textit{solution manifold} $M \subseteq \mathcal{J}^{(1)}$ is then those $(t,y,\dot{y})$ such that $\dot{y} = \omega(t,y)$.
        
        In our context, where $E \cong \mathbb{R}^2$, we seek sets of transformations $\Gamma_\varepsilon \colon (t,y) \mapsto (\hat{t},\hat{y})$, depending on a transformation parameter $\varepsilon$, satisfying the following conditions:
        \begin{enumerate}[start=1,label={(C\arabic*)}]
            \item \label{diffeomorphisms} $\Gamma_\varepsilon$ are diffeomorphisms of $E$,
            \item \label{symmetry condition} $\frac{d\hat{y}}{d\hat{t}}=\omega(\hat{t},\hat{y})$ whenever $\dot{y} = \frac{dy}{dt}=\omega(t,y)$,
            \item \label{analytic} $\hat{t}(t,y;\varepsilon)$ and $\hat{y}(t,y;\varepsilon)$ are analytic functions of $t$ and $y$ in a neighbourhood $\mathcal{U}$ of $\varepsilon = 0$, where $\Gamma_0 = \mathrm{id}$, and $\Gamma_\varepsilon \Gamma_\delta = \Gamma_{\varepsilon + \delta}$ for all $\varepsilon,\delta \in \mathcal{U}$.
        \end{enumerate}

        In general, we say a local group of transformations, $\Gamma$, is a \textit{symmetry} of the ODE if it acts on open subsets of $E$ by diffeomorphisms that, when extended via \ref{symmetry condition} to $\mathcal{J}^{(1)}$, restrict to bijections on $M$. Informally, we can think of symmetries to be transformations that map solutions to solutions.

        The group of transformations satisfying \ref{diffeomorphisms}, \ref{symmetry condition} and \ref{analytic} forms (a representation of) a local one-parameter Lie group; we will refer to such $\Gamma_\varepsilon$ as \textit{Lie symmetries}.
    \end{defn}
    
    \begin{defn}
        Given a Lie symmetry, we define the \textit{infinitesimals} $\xi(t,y)$ and $\eta(t,y)$ as the coefficients of $\varepsilon$ in the expansion of $\hat{t}$ and $\hat{y}$:
        \begin{align*}
            \hat{t}(t,y;\varepsilon) = t + \xi(t,y)\varepsilon + \mathcal{O}(\varepsilon^2), \\
            \hat{y}(t,y;\varepsilon) = y + \eta(t,y)\varepsilon + \mathcal{O}(\varepsilon^2).
        \end{align*}
        These define a smooth vector field $X$, on the total space $E$, referred to as \textit{the infinitesimal generator of the Lie group}, given by
        \[
            X = \xi(t,y)\frac{\partial}{\partial t} + \eta(t,y)\frac{\partial}{\partial y}.
        \]
    \end{defn}

    \paragraph{}
    This infinitesimal generator completely characterises $\Gamma_\varepsilon$ through the following theorem, due to Lie \cite{Lie1, Lie2, Lie3, Hyd}.

    \begin{thm}
    \label{1FT} 
        $\hat{t}$ and $\hat{y}$ can be locally recovered by solving the differential equations
        \begin{align*}
            \frac{d\hat{t}}{d\varepsilon} &= \xi(\hat{t},\hat{y}) & \text{with initial condition } \hat{t}(t,y;0) = t, \\
            \frac{d\hat{y}}{d\varepsilon} &= \eta(\hat{t},\hat{y}) & \text{with initial condition } \hat{y}(t,y;0) = y.
        \end{align*}
    \end{thm}
    
    \paragraph{}
    To find $\Gamma_\varepsilon$, it therefore remains to find expressions for $\xi$ and $\eta$, which we do by solving the \textit{linearised symmetry condition}, a PDE in $\xi(t,y)$ and $\eta(t,y)$ which, for first-order ODE models $\dot y = \omega(t,y)$, is given by the following equation \cite{Hyd}:
        \[
            \eta_t + (\eta_y - \xi_t)\omega - \xi_y \omega^2 = \xi \omega_t + \eta \omega_y. \tag{LSC}
        \]

    \begin{defn}
        Typically, for higher order ODE models, the corresponding linearised symmetry condition is solved by noting both $\xi$ and $\eta$ are formally independent of $\dot{y}$ and higher derivatives of $y$, so comparing coefficients allows us to decompose the linearised symmetry condition into a system of PDEs known as the \textit{determining equations}.
    \end{defn}

    \begin{rem}
        It is worth noting that specific $\xi$ and $\eta$ that satisfy a given linearised symmetry condition can also be found using various ans{\"a}tze, and that a general solution to the determining equations is not always necessary to carry out the analysis that follows. However, as we will show later, more freedom in choosing parameters of the symmetries can lead to better differentiation between models.
    \end{rem}

    \paragraph{}
    We illustrate this by considering the example of the logistic growth model, whose symmetries we will later generalise to those of the Richards model.

    \begin{exam}
    \label{logistic symmetry}
        Consider the logistic growth model given by
        \begin{align*}
            \dot{y} &= \omega(t,y) = r_{\mathrm{L}}y\left(1-\frac{y}{K}\right), & \text{with solution } y(t) = \frac{K}{1+\left(\frac{K}{y_0} - 1\right)e^{-r_{\mathrm{L}}t}}, 
        \end{align*}
        where $r_{\mathrm{L}}$ is the intrinsic growth rate, $K$ the carrying capacity, and $y_0$ the initial population. Noting the system is autonomous, we try using the ansatz $\xi(t,y) = t$, so that the linearised symmetry condition (LSC) reads
        \[
            \eta_t + \eta_y\dot{y} = \left(r_{\mathrm{L}} - \frac{2r_{\mathrm{L}}y}{K}\right) \eta + \dot{y}.
        \]
        In this case, the solution can be found using the method of characteristics, and gives
        \[
            \eta(t,y) = r_{\mathrm{L}}y\left(1-\frac{y}{K}\right)t + {c_1} y^2 e^{-r_{\mathrm{L}}t},
        \]
        where ${c_1}$ is an arbitrary integration constant. We are left with the infinitesimal generator
        \[
            X = t\frac{\partial}{\partial t} + \left(r_{\mathrm{L}}y\left(1-\frac{y}{K}\right)t + {c_1} y^2 e^{-r_{\mathrm{L}}t}\right)\frac{\partial}{\partial y}.
        \]
        Applying Theorem \ref{1FT}, we can recover $\Gamma_\varepsilon$ by solving the ODEs
        \begin{align}
            \label{t^ L}\frac{d\hat{t}}{d\varepsilon} &= \hat{t} & \text{with } \hat{t}(t,y;0) = t, \\
            \label{y^ L} \frac{d\hat{y}}{d\varepsilon} &= r_{\mathrm{L}}\hat{y}\left(1-\frac{\hat{y}}{K}\right)\hat{t} + c_1\hat{y}^2e^{-r_{\mathrm{L}}\hat{t}} & \text{with } \hat{y}(t,y;0) = y.
        \end{align}
        Equation (\ref{t^ L}) has solution $\hat{t} = te^\varepsilon$, making equation (\ref{y^ L}) separable for $c_1=0$, which gives the solution $\hat{y} = K\left[1+\left(\frac{K}{y} - 1\right)\operatorname{exp}\{-r_{\mathrm{L}} t (e^{\varepsilon}-1)\}\right]^{-1}$. We hence obtain the \textit{logistic symmetry}, $\Gamma_\varepsilon^{(\mathrm{L})}$, given by
        \[
            \Gamma_\varepsilon^{(\mathrm{L})}(t,y) = \left(te^\varepsilon, K\left[1+\left(\frac{K}{y} - 1\right)\operatorname{exp}\{-r_{\mathrm{L}} t (e^{\varepsilon}-1)\}\right]^{-1}\right).
        \]
    \end{exam}

    \paragraph{}
    In fact, the idea used for the logistic growth model in Example \ref{logistic symmetry} of employing an ansatz of the form $\xi = A(t)$, for an arbitrary function $A$, to reduce the problem to solving a semi-linear PDE (using the method of characteristics) also works for any first-order autonomous ODE whose solution is known. In particular, every differential equation $\dot{y} = \omega(y)$ has a symmetry generated by $\xi = A(t)$, $\eta = A(t)\omega(y)$. Such a Lie symmetry is referred to as a \textit{trivial symmetry} (as it is inherent to any autonomous model, even if the associated transformation $\Gamma$ acts nontrivially on the total space).

\section{Models and Symmetries}
    \paragraph{}
    We start by giving a summary of the three models of interest, then derive some of the Lie symmetries used for model selection in Section \ref{results}.
\subsection{Models of Interest}
    \begin{model}[Logistic growth model]
        The first of our three models has already been discussed in Example \ref{logistic symmetry}, where we derived a family of transformations $\Gamma_\varepsilon^{(\mathrm{L})}$.
    \end{model}
    
    \begin{model}[Gompertz model]
        The second of our models has a few formulations, of which we present two. The \textit{autonomous} formulation is given by
        \[
            \dot{y} = r_{\mathrm{G}}y\operatorname{log}\left(\frac{K}{y}\right), \label{autonomous}\tag{A}
        \]
        while the \textit{classical} formulation is given by 
        \[
            \dot{y} = r_{\mathrm{G}}y\operatorname{log}\left(\frac{K}{y_0}\right)e^{-r_{\mathrm{G}}t}. \label{classical}\tag{C}
        \]
        It is important to note that fixing $y(0) = y_0$ in both formulations gives the same general solution
        \[
            y(t) = K \operatorname{exp}\left\{e^{-r_{\mathrm{G}}t} \operatorname{log}\left(\frac{y_0}{K} \right)\right\},
        \]
        and that we will primarily be working with the autonomous formulation, (\ref{autonomous}), although symmetries of the classical Gompertz model have also been considered in \cite{Aug}.
    \end{model}

    \begin{model}[Richards model]
        Our final model of interest is given by
        \begin{align*}
            \dot{y} &= r_{\mathrm{R}}y\left(1 - \left(\frac{y}{K}\right)^{\beta}\right), & \text{with solution } y(t) = \frac{K}{\left[1+\left(\frac{K^{\beta}}{y_0^{\beta}} - 1\right)e^{-\beta r_{\mathrm{R}}t}\right]^{\frac{1}{\beta}}}.
        \end{align*}
    \end{model}

    \begin{rem}
    \label{fitting}
        We explicitly differentiate between the three intrinsic growth rates $r_{\mathrm{L}}$, $r_{\mathrm{G}}$, and $r_{\mathrm{R}}$, while keeping the same notation for the carrying capacities $K$ and the initial populations $y_0$. The importance of distinguishing the growth rates becomes apparent when drawing conclusions from models fitted to data. For example, the paper by Simpson et. al. gives a comparison of $\frac{1}{r}$, representing the approximate \textit{regrowth timescales}, when fitting all three models to some data in \cite{Sim}. Using the Lady Musgrave Reef data from \cite{War}, and fitting with ordinary least squares (OLS), we see that the Gompertz, logistic and Richards models predict timescales of 712, 400 and 259 days, respectively. We note that all fitting in the paper is done via OLS, using \texttt{SciPy}'s \texttt{curve\_fit}.
    \end{rem}

\subsection{Symmetries of Considered Models}
    \paragraph{}
    We seek a Lie symmetry of the Gompertz model by employing the same strategy as Example \ref{logistic symmetry}, and use the generators $\xi = t$, $\eta = r_{\mathrm{G}}y\operatorname{log}\left(\frac{K}{y}\right)t$ again.
    
    Next, we find a group of symmetries $\Gamma_\varepsilon^{\mathrm{(G)}}(t,y) = (\hat{t},\hat{y})$ by considering 
    \begin{align}
        \label{t^ G}\frac{d\hat{t}}{d\varepsilon} &= \hat{t} & \text{with } \hat{t}(t,y;0) = t, \\
        \label{y^ G} \frac{d\hat{y}}{d\varepsilon} &= \hat{t}r_{\mathrm{G}}\hat{y}\operatorname{log}\left(\frac{K}{\hat{y}}\right) & \text{with } \hat{y}(t,y;0) = y.
    \end{align}
    
    We will refer to the solution $\left(\hat{t},\hat{y}\right)$ as the \textit{Gompertz symmetry}, $\Gamma_\varepsilon^{\mathrm{(G)}}$, given by
    \[
        \Gamma_\varepsilon^{\mathrm{(G)}}(t,y) = \left(te^\varepsilon, K\left(\frac{y}{K}\right)^{\operatorname{exp}\left(-r_{\mathrm{G}}t\left(e^\varepsilon-1\right)\right)}\right).
    \]
    
    \begin{exam}
    \label{Felix ansatz} 
        
        Many other symmetries of the autonomous Gompertz model (\ref{autonomous}) have been found in \cite{Aug} by considering ans{\"a}tze of the form
        \begin{align*}
            \xi(t,y) &= A(t) + B(t)\operatorname{log}\left(\frac{y}{K}\right), \\
            \eta(t,y) &= C(t) + D(t)y\operatorname{log}\left(\frac{y}{K}\right),
        \end{align*}
        and solving the determining equations. This yielded a Lie algebra of symmetries spanned by
        \begin{align*}
            X_1 &= e^{r_{\mathrm{G}}t}\operatorname{log}\left(\frac{y}{K}\right)\frac{\partial}{\partial t}, \\
            X_2 &= e^{-r_{\mathrm{G}}t}y\frac{\partial}{\partial y}, \\
            X_3 &= y\operatorname{log}\left(\frac{y}{K}\right)\frac{\partial}{\partial y}, \text{ and}\\
            X_4 &= A(t)\frac{\partial}{\partial t} + r_{\mathrm{G}}A(t)y\operatorname{log}\left(\frac{y}{K}\right)\frac{\partial}{\partial y},
        \end{align*}
        where $A$ is an arbitrary function of $t$.
    \end{exam}

    \paragraph{}
    The final symmetry we require before tackling the model selection problem is a symmetry of the Richards model. Since the Richards model also has an exact solution, we can repeat the argument from the logistic growth model to get
    \[
        X = t\frac{\partial}{\partial t} + \left(r_{\mathrm{R}}y\left(1-\left(\frac{y}{K}\right)^\beta\right)t + {c_1} y^{\beta+1} e^{-\beta r_{\mathrm{R}}t}\right)\frac{\partial}{\partial y}.
    \]
    We now solve the following to find a group of symmetries $\Gamma_\varepsilon^{(\mathrm{R})}(t,y) = (\hat{t},\hat{y})$ in the case where $c_1 = 0$:
    \begin{align}
        \label{t^ R}\frac{d\hat{t}}{d\varepsilon} &= \hat{t} & \text{with } \hat{t}(t,y;0) = t, \\
        \label{y^ R} \frac{d\hat{y}}{d\varepsilon} &= r_{\mathrm{R}}\hat{y}\left(1-\left(\frac{\hat{y}}{K}\right)^\beta\right)\hat{t} & \text{with } \hat{y}(t,y;0) = y.
    \end{align}
    Solving these, we obtain the \textit{Richards symmetry} as the solution, $\Gamma_\varepsilon^{\mathrm{(R)}}$, given by
    \[
        \Gamma_\varepsilon^{\mathrm{(R)}}(t,y) = \left(te^\varepsilon, K\left[1+\left(\left(\frac{K}{y}\right)^\beta - 1\right)\operatorname{exp}\{-\beta r_{\mathrm{R}} t (e^{\varepsilon}-1)\}\right]^{-\frac{1}{\beta}}\right).
    \] 
    \paragraph{}
    Other Lie symmetries can be found by means of other ans{\"a}tze, such as generalisations of the ones from Example \ref{Felix ansatz}, or generalisations of the ones proposed by Cheb-Terrab--Kolokolnikov in \cite{CTK}. When working with higher-order ODEs, or with systems of ODEs, it is possible to obtain an overdetermined set of determining equations, giving a complete set of infinitesimal generators for the Lie algebra of symmetries and avoiding the need for ans{\"a}tze altogether. Another method for finding symmetries is discussed in Appendix \ref{appendix}.

\section{Results}
\label{results}
\subsection{A Symmetry-based Framework for Model Selection}
    \paragraph{}
    We begin this section by justifying the symmetry-based framework for model selection, as explained by Ohlsson et al. in \cite{OBC}. We observe that, by construction, the solution space to a model is closed under the action of a Lie symmetry of that model, but not necessarily under other transformations (e.g. Lie symmetries of other models). We therefore consider the following algorithm, adapted from \cite{Bor}.

    \begin{alg}
    \label{algorithm}
        We initialise the model selection algorithm by inputting: time series data, $(t_i,y_i)_{i=1}^N$; and a Lie symmetry, $\Gamma_\varepsilon$, of a candidate differential equation model, with solution curve $m$.
        \begin{enumerate}
            \item Choose a sufficiently small\footnote{See Remark \ref{how small} and Question \ref{Q1}.} transformation parameter, $\varepsilon$, and transform the data via $\Gamma_\varepsilon$ to obtain transformed time series data $(\hat{t}_i,\hat{y}_i)_{i=1}^N = \Gamma_\varepsilon(t_i,y_i)_{i=1}^N$.
            \item Fit\footnote{As mentioned in Remark \ref{fitting}, we fit using OLS.} a curve $\hat{m}(\hat{t})$ of the candidate model to the transformed data.
            \item Using $\Gamma_{-\varepsilon}$, inversely transform $\hat{m}$ to obtain a new curve $(t,\widecheck{m}(t)) = \Gamma_{-\varepsilon}(\hat{t},\hat{m}(\hat{t}))$.
        \end{enumerate}
        The algorithm outputs a curve $\widecheck{m}(t;\varepsilon)$, which we will analyse to either support or refute whether the candidate model is (structurally) a good fit to the data.
    \end{alg}

    \paragraph{}
    \label{low noise}
    As we mentioned above, a Lie symmetry, $\Gamma_\varepsilon$, of a model with solution curve $m$ will map solutions to solutions. Hence, given generated time series data $(t_i, m(t_i)+\nu_i)$ for some noise $\nu_i$ drawn from a distribution of mean $0$ and variance ${\sigma_{\nu}}^2$, we expect $\widecheck{m}(t;\varepsilon)$ to be an excellent fit to the data in a neighbourhood of $\sigma_\nu = 0$.
    We will refer to this neighbourhood of $\sigma_\nu=0$ as the \textit{low-noise regime}.

    In light of this, we present some preliminary results, using simulated data, that exhibit the success of the trivial symmetry in distinguishing the correct underlying models in the low-noise regime. We then apply the same framework to coral reef data from the Great Barrier Reef.

    \begin{rem}
    \label{how small}
        As we see in Question \ref{Q1} addressing the neighbourhood of $\varepsilon=0$, and Question \ref{Q2} addressing the neighbourhood of $\sigma_\nu=0$, the question of ``how small is \textit{sufficiently small}?" remains open for both. In this paper, we chose $\sigma_\nu = 1$ and $\varepsilon$ to be some $0 < \varepsilon < 1$ where the quality of fit (as measured by the coefficient of determination, $R^2$) of a correct underlying model retains a good fit after implementing Algorithm \ref{algorithm} for simulated data. This motivates generating simulated data that is faithful to the real data considered, as this $\varepsilon$ can then be reused as an upper bound for the transformation parameter when using real data.
    \end{rem}
    
\subsection{Generating Convincing Simulated Data}
    \paragraph{}
    The \textit{reef dataset} is a rich time series dataset of 120 coral reef sites from around the Great Barrier Reef collected by the Australian Institute of Marine Science (AIMS), a processed version of which is available due to \cite{War}. To simulate realistic data, we analyse the processed data to find some summary statistics. The data of interest are:
    \begin{enumerate}
        \item \texttt{Date}, the dates of visits to the reef site,
        \item \texttt{HC}, hard coral cover, which is an aggregate statistic representing the populations of hard corals (various \textit{Acroporidae}, non-\textit{Acroporidae} hard coral, \textit{Fungiidae} and solitary hard coral),
        \item \texttt{HC\_sd}, the (sample, as each \texttt{HC} observation is an average of five observations from different transects) standard deviation of each \texttt{HC} observation.
    \end{enumerate}

    The first step in generating convincing data is obtaining a distribution of the times between observations. The Lady Musgrave Reef data analysed in \cite{Sim} and in \ref{concept} was recorded between 1992 and 2003. The frequency of sampling at all sites before 2006 was roughly annual, with sample mean $\mu = 366.5$ days and sample variance $\sigma^2 = 57.12^2$, which is consistent with the sampling frequency guidelines published by AIMS at the time. Note the guidelines were changed to recording every other year, after 2006.

    While this may be sufficient to simulate visually convincing coral cover data, the reef dataset also includes \texttt{HC\_sd}, which helps optimise the fitting by quantifying the uncertainty in the data. To simulate this uncertainty, we therefore consider the standard deviation \texttt{HC\_sd} against the coral cover \texttt{HC}, and observe significant heteroscedasticity, so we fit an order 2 polynomial to the standard deviation; the residual plot below reveals this to be reasonable.

    \begin{figure}[ht]
        \begin{subfigure}{0.48\textwidth}
            \includegraphics[width=\linewidth]{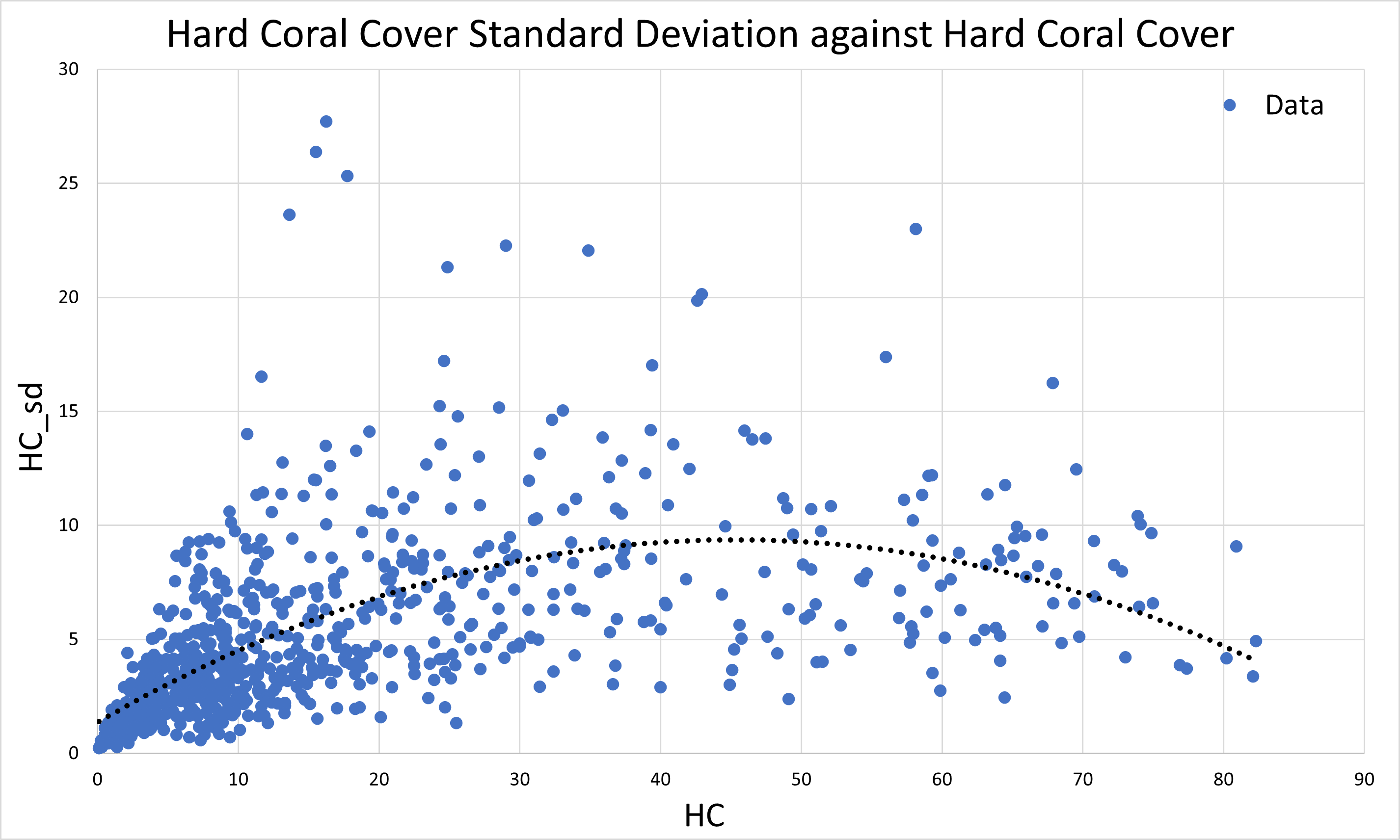}
            \caption{Heteroscedasticity in \texttt{HC\_sd}} \label{fig:1a}
        \end{subfigure}\hspace*{\fill}
        \begin{subfigure}{0.48\textwidth}
            \includegraphics[width=\linewidth]{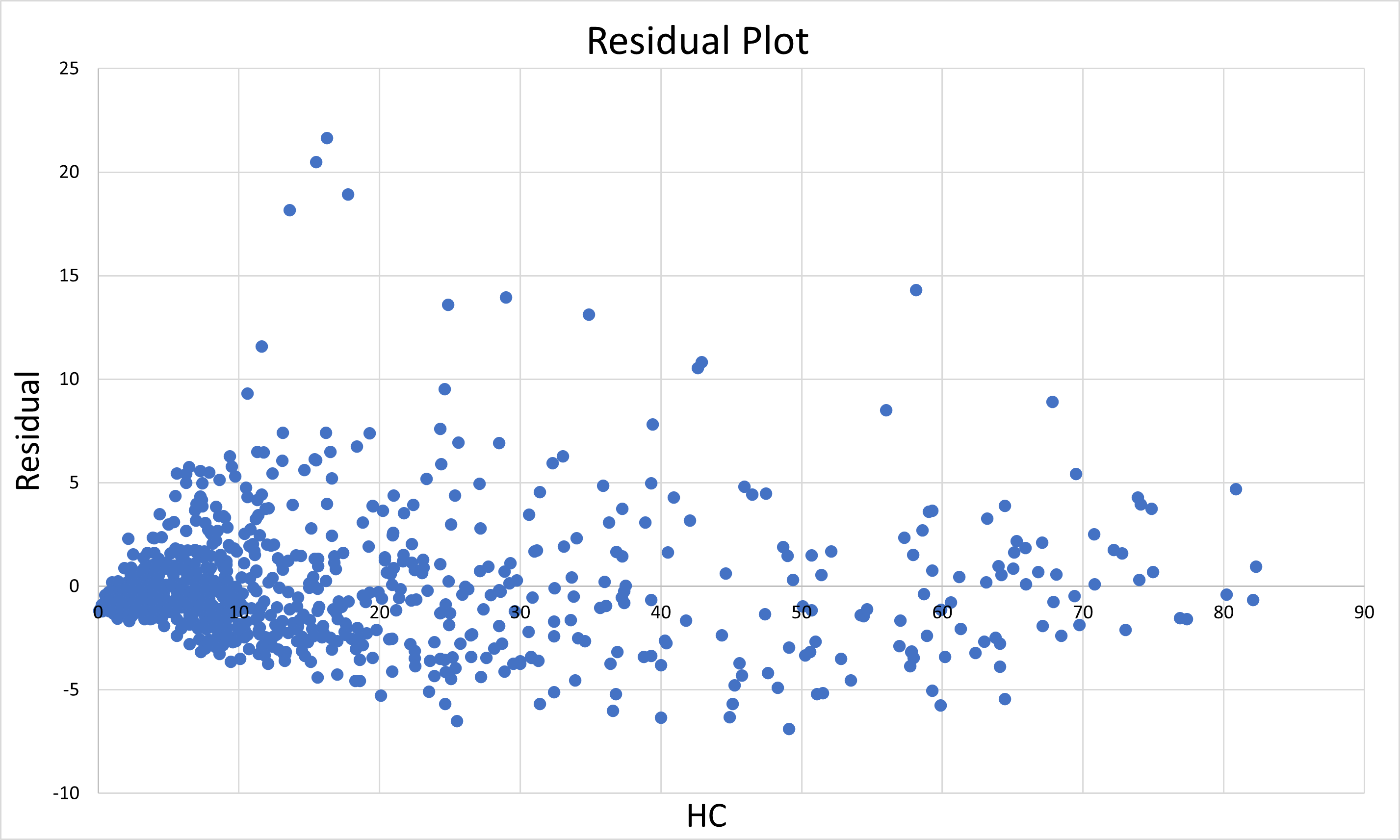}
            \caption{Residual Plot} \label{fig:1b}
        \end{subfigure}
        \caption{Coral Cover Standard Deviation (\texttt{HC\_sd}) as Coral Cover (\texttt{HC}) Varies}
    \end{figure}

    Finally, we must decide on parameters $r, K, y_0$ and $\beta$ with which to generate data. We fix $\beta=\frac{1}{2}$, as $\hat\beta=\frac{1}{2}$ is the mode value of $\hat\beta$ to the nearest half-integer when fitting the Richards model (with $\beta$ free) to all sites of the reef dataset. $\beta = \frac{1}{2}$ also interpolates between the logistic model (where $\beta = 1$) and the Gompertz model (which is the limit of the Richards model as $\beta r_\mathrm{R} \to r_\mathrm{G}$ and $\beta \to 0$). We pick the other parameters as the means of the fitted parameters across three considered datasets to allow a direct comparison. The three chosen reefs (Lady Musgrave Reef 1992-2003, One Tree Reef 1992-2003, Broomfield Reef 2008-2018) were similar in terms of number of samples and sampling frequency. Furthermore, these parameter values are used as priors for fitting.

\subsection{Proof of Concept -- the Trivial Symmetry Solves the Model Selection Problem for Simulated Data}
\label{proof of concept}
    \paragraph{} 
    We present some figures of the fitted curves $\hat{m}$, which are referred to in the figures as the \textit{old fits} (dashed lines), and $\widecheck{m}$, which are referred to as the \textit{transformed fits} (solid lines). The data below is simulated from the three models of interest with added Gaussian noise (with mean 0 and standard deviation 1).

    \begin{figure}[ht]
        \begin{subfigure}{0.48\textwidth}
            \includegraphics[width=\linewidth]{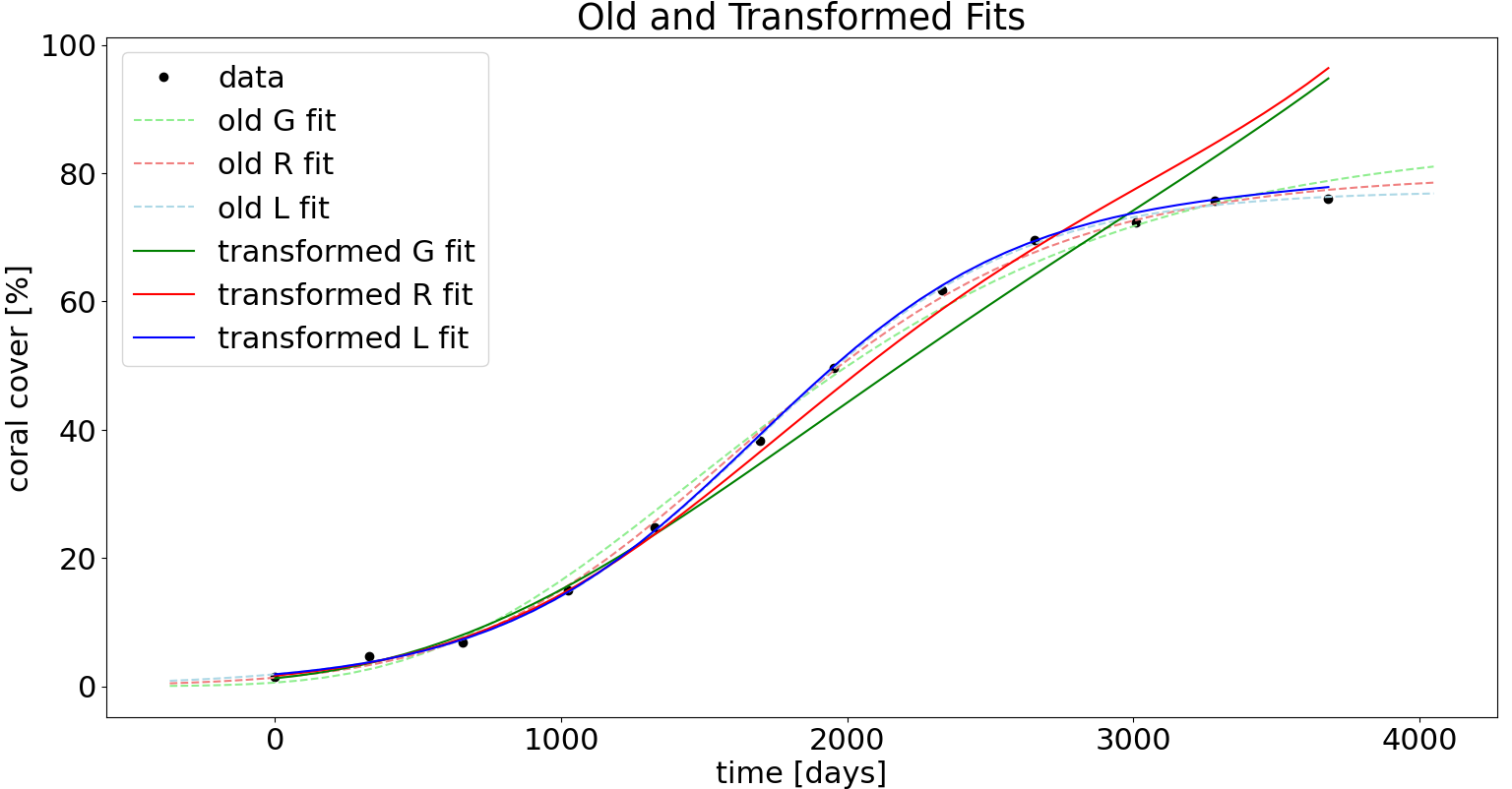}
            \caption{Simulated Logistic Data} \label{fig:2a}
        \end{subfigure}\hspace*{\fill}
        \begin{subfigure}{0.48\textwidth}
            \includegraphics[width=\linewidth]{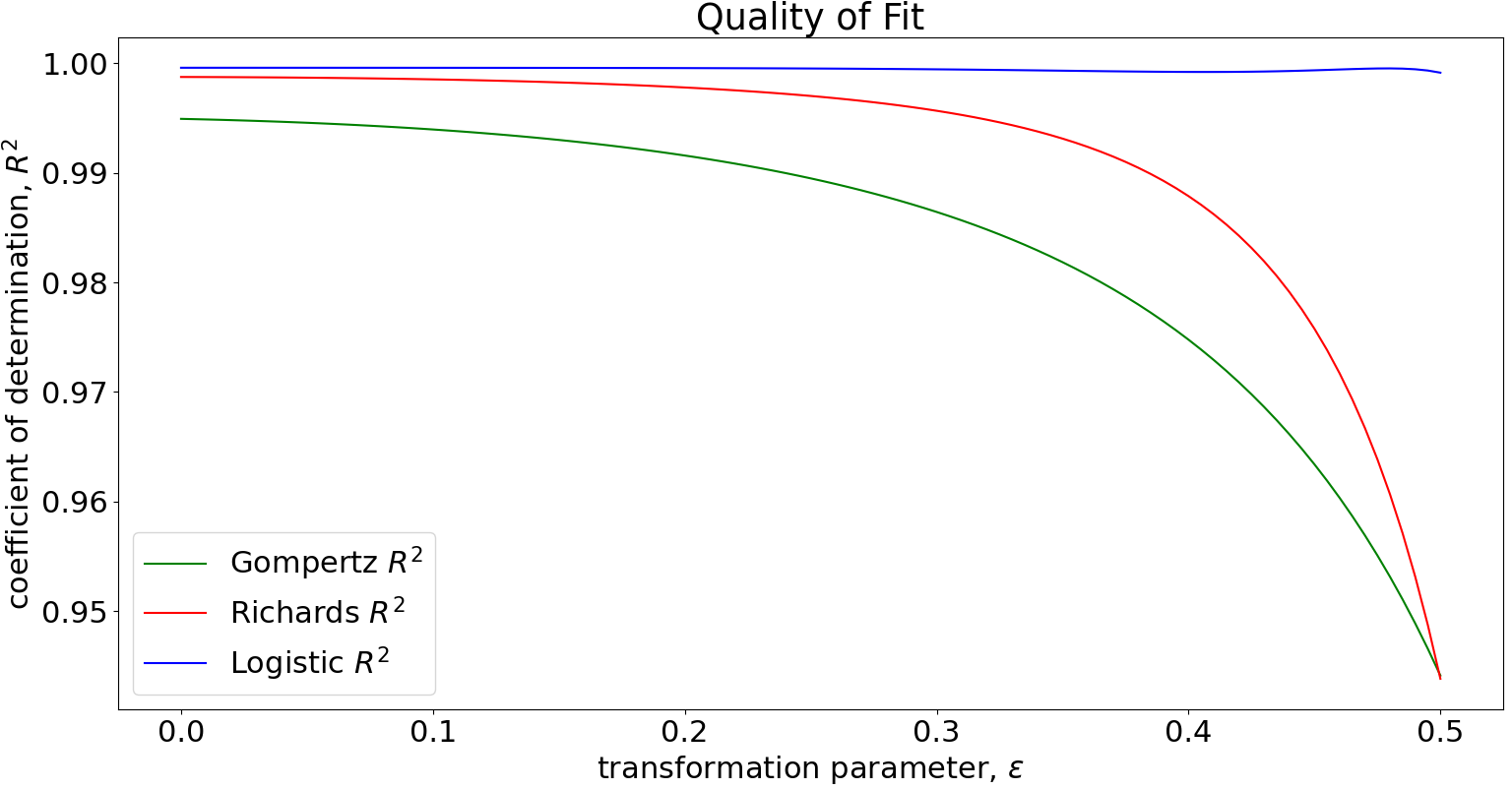}
            \caption{Quality of Fit} \label{fig:2b}
        \end{subfigure}
    \end{figure}
    \begin{figure}[ht]\ContinuedFloat
        \begin{subfigure}{0.48\textwidth}
            \includegraphics[width=\linewidth]{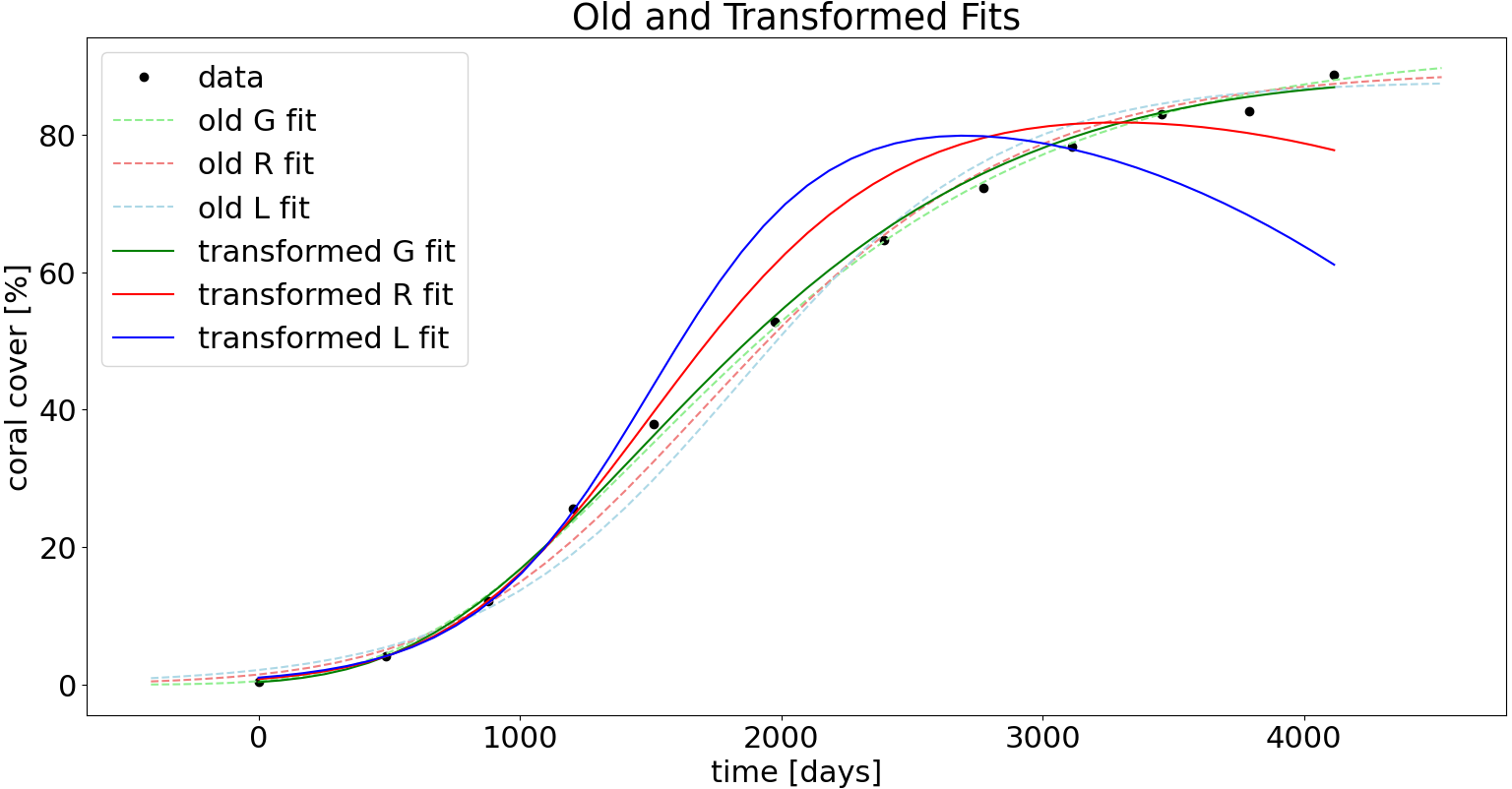}
            \caption{Simulated Gompertz Data} \label{fig:2c}
        \end{subfigure}\hspace*{\fill}
        \begin{subfigure}{0.48\textwidth}
            \includegraphics[width=\linewidth]{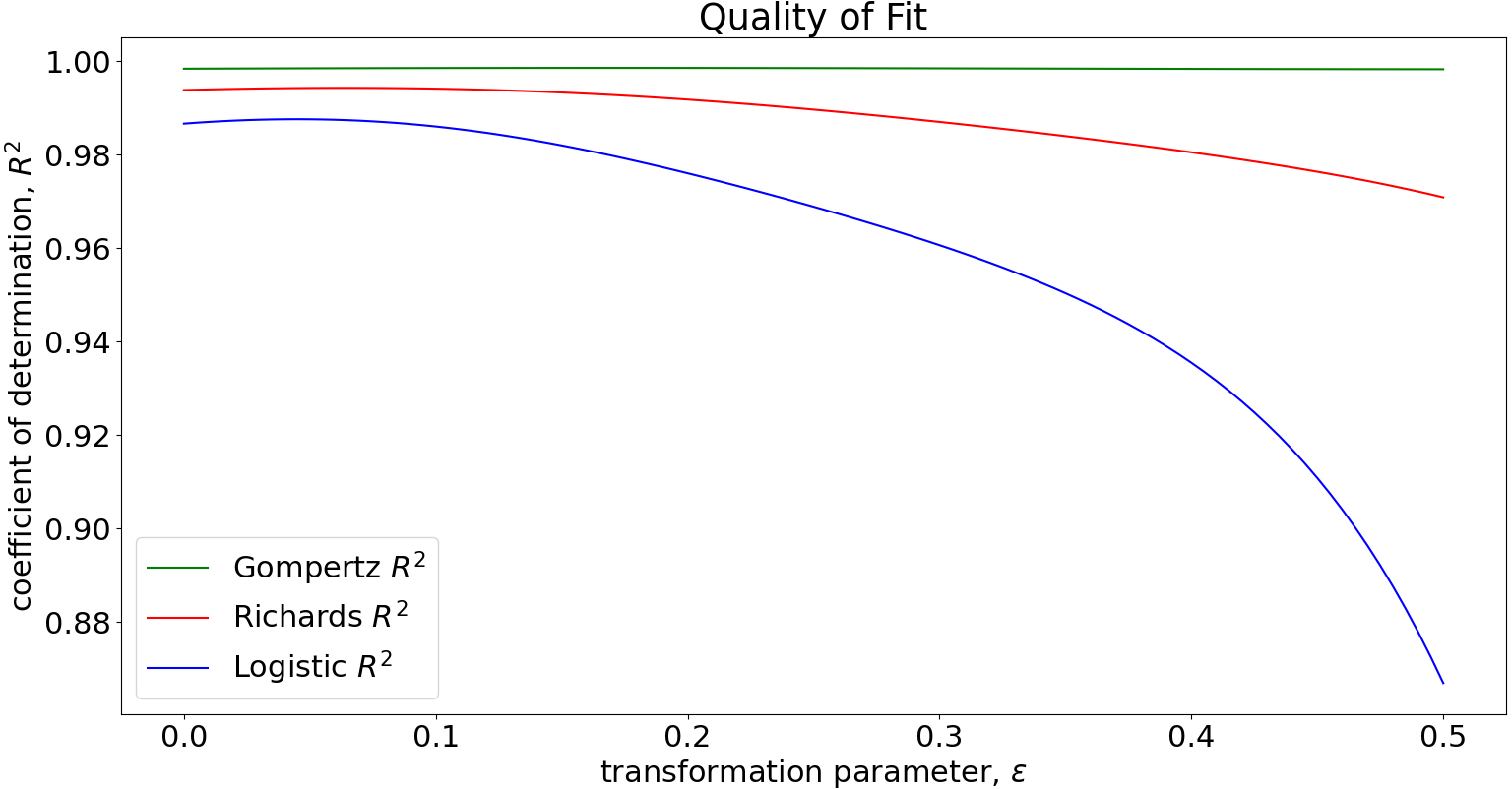}
            \caption{Quality of Fit} \label{fig:2d}
        \end{subfigure}
    \end{figure}
    \begin{figure}[ht]\ContinuedFloat
        \begin{subfigure}{0.48\textwidth}
            \includegraphics[width=\linewidth]{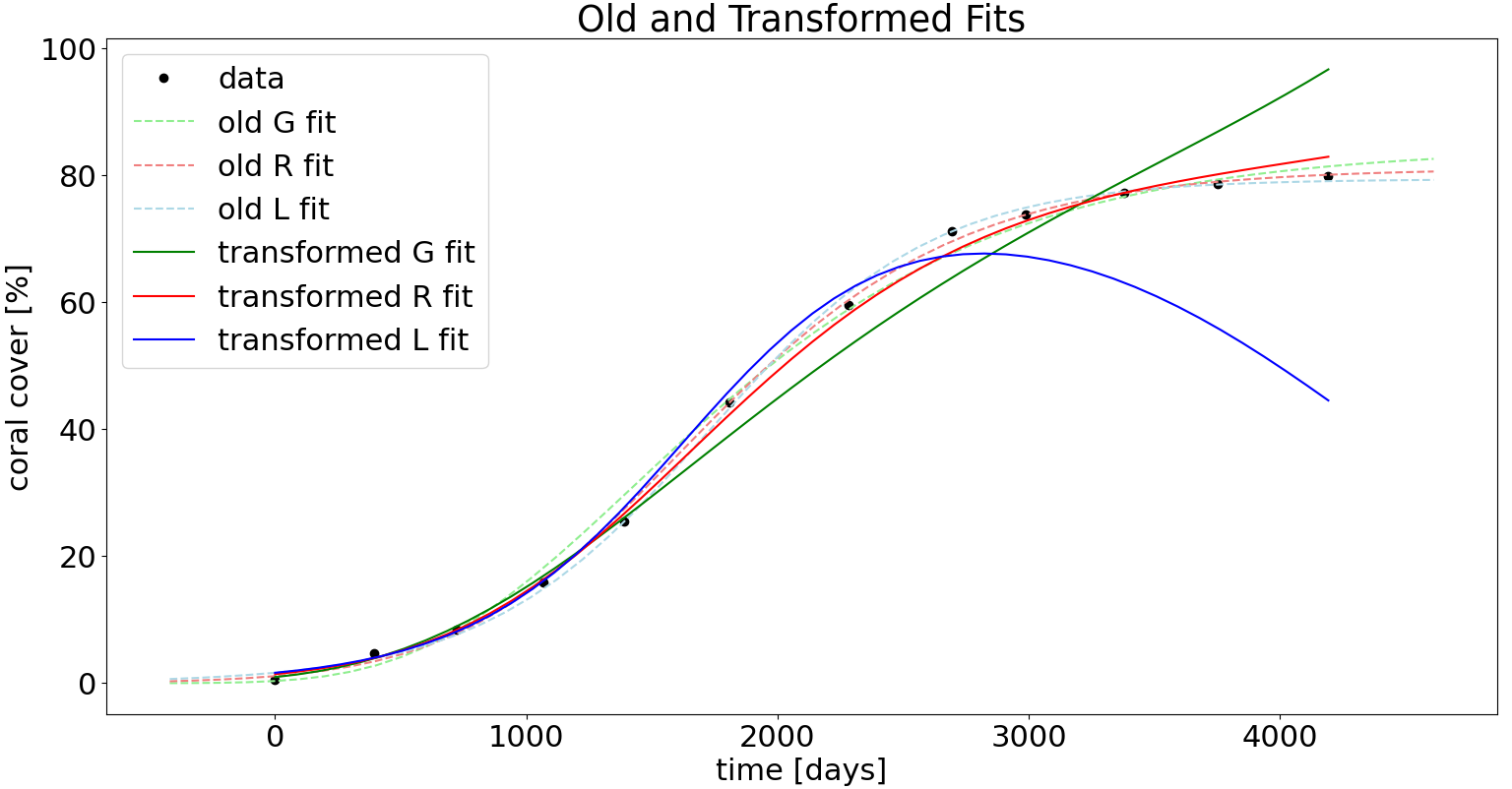}
            \caption{Simulated Richards Data} \label{fig:2e}
        \end{subfigure}\hspace*{\fill}
        \begin{subfigure}{0.48\textwidth}
            \includegraphics[width=\linewidth]{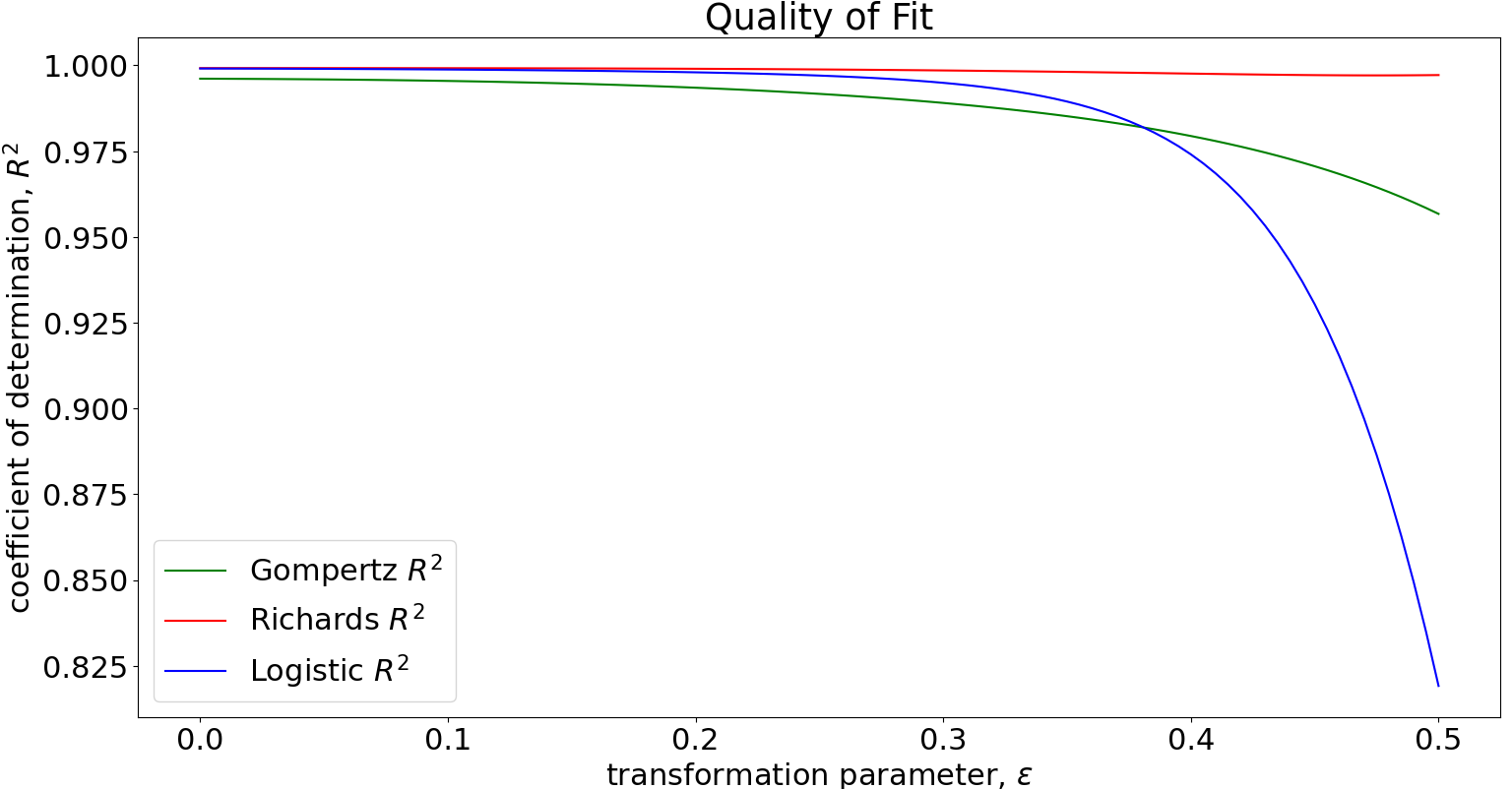}
            \caption{Quality of Fit} \label{fig:2f}
        \end{subfigure}
        \caption{Simulated Data and Quality of Fit of Fitted Models}\label{fig:2}
    \end{figure}

    As shown by the figures above, these preliminary results suggest that our framework for model selection is capable of correctly identifying the correct underlying model, and that quality of fit decreases as a function of the transformation parameter, $\varepsilon$. It is expected that the underlying model's quality of fit should also decrease as the noise in the simulated data increases. 

\subsection{Concept -- Tackling the Model Selection Problem for Real Data}
\label{concept}
    \paragraph{}
    In light of the proof of concept in \ref{proof of concept}, we now apply the same transformations to data from the Lady Musgrave Island from the Great Barrier Reef.
    
    \begin{figure}[ht]
        \begin{subfigure}{0.48\textwidth}
            \includegraphics[width=\linewidth]{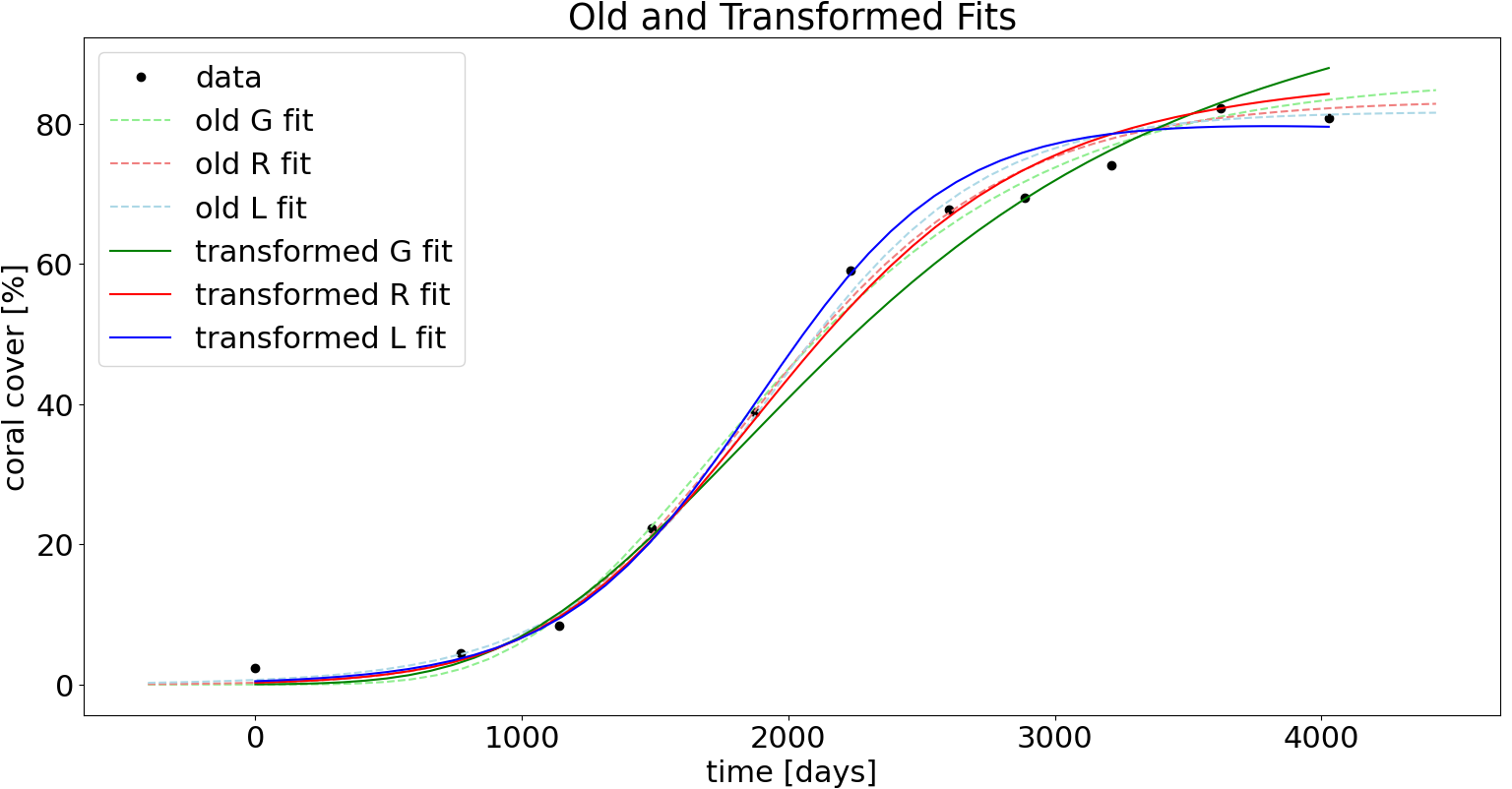}
            \caption{Lady Musgrave Reef Data} \label{fig:3a}
        \end{subfigure}\hspace*{\fill}
        \begin{subfigure}{0.48\textwidth}
            \includegraphics[width=\linewidth]{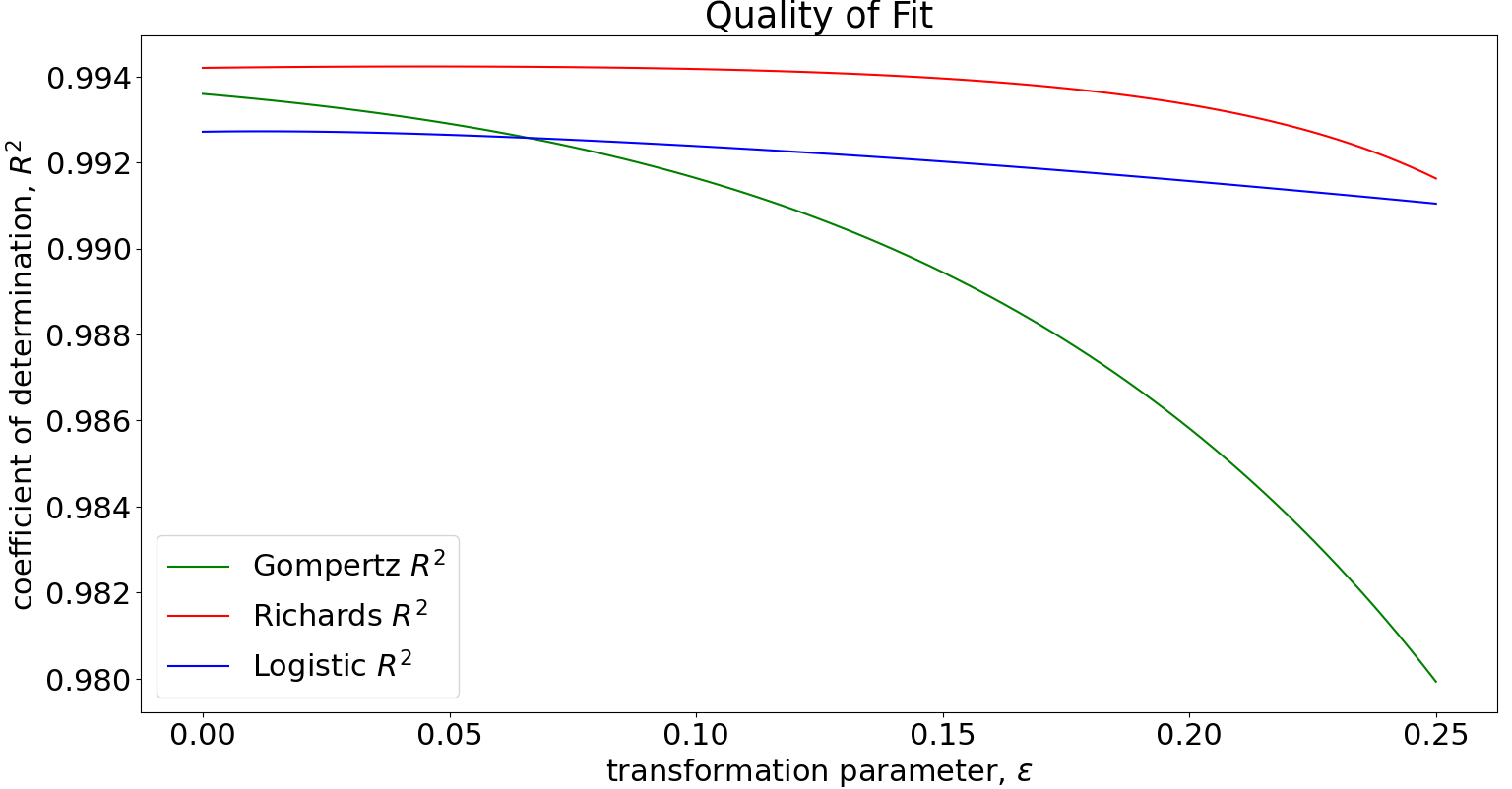}
            \caption{Quality of Fit}
            \label{fig:3b}
        \end{subfigure}
        \caption{The Framework Applied to Real Data Favours the Logistic and Richards}
        \label{fig:3}
    \end{figure}

    The figures above support the use of the logistic or Richards models over the Gompertz model: the deterioration in quality of fit (as $\varepsilon$ increases) is highest for the Gompertz model out of the three candidates. At the very least, this evidence suggests that a careful evaluation of model validity is necessary before using any particular model to draw conclusions from data. This is especially relevant in coral reef modelling given the wide usage of the Gompertz model \cite{Osb,Ive,Tho}.

\subsection{Determining Optimal Parameters for Model Selection -- The Disagreement Coefficient}
\label{Delta subsection}

    \paragraph{}
    We begin this subsection by remarking that, so far, all the transformations $\Gamma$ we have used have had one parameter, $\varepsilon$. We are also interested in cases where the symmetries have multiple parameters. For example, since the set of infinitesimal generators form a vector space, we can consider $X = c_1X_1 + \dots + c_nX_n$ for a linearly independent set of generators $\{X_1, \dots, X_n \}$ to obtain a transformation that depends on a vector of parameters $\mathbf{c} = \left(c_1, \dots, c_n\right)$. To study the dependence of the quality of fit on the parameters of the transformation, we introduce the \textit{disagreement coefficient}.

    To do so, we first choose symmetries of the models of interest that depend on a transformation parameter, $c$. Using the ansatz $\xi = 1$, we obtain the following symmetries by solving for $\eta$ using the method of characteristics, and using Theorem \ref{1FT} to find 
    \begin{align*}
        \Pi_{\varepsilon}^{(\mathrm{L})}(t,y;c) &= \left(t+\varepsilon, \left[\frac{1}{y} + \frac{c}{r_{\mathrm{L}}}e^{-r_{\mathrm{L}}t}\left(e^{-r_{\mathrm{L}}\varepsilon} - 1\right)\right]^{-1}\right), \\
        \Pi_{\varepsilon}^{(\mathrm{G})}(t,y;c) &= \left(t+\varepsilon, y\operatorname{exp}\left\{\frac{c}{r_{\mathrm{G}}}e^{-r_{\mathrm{G}}t}\left(1 - e^{-r_{\mathrm{G}}\varepsilon}\right)\right\}\right), \text{ and}\\
        \Pi_{\varepsilon}^{(\mathrm{R})}(t,y;c) &= \left(t+\varepsilon, \left[y^{-\beta} + \frac{c}{r_{\mathrm{R}}}e^{-\beta r_{\mathrm{R}}t}\left(e^{-\beta r_{\mathrm{R}}\varepsilon} - 1\right)\right]^{-\frac{1}{\beta}}\right).
    \end{align*}

    \paragraph{}
    We now seek optimal values of $c$ by maximising the \textit{disagreement} between two models. To define this properly, we need some notation.

    \begin{nota}
        Henceforth when considering a model, A, given by $\dot{y}=\omega(t,y)$, we will use a lowercase letter $a = a(t)$ to denote the general solution, and $\hat{a} = \hat{a}(t;d)$ to denote a fitted (via OLS) solution given some data $d = (t_i,y_i)_{i=1}^N$. This is the the ``old fit" in Figures \ref{fig:2} and \ref{fig:3}.

        Unravelling Algorithm \ref{algorithm}, we can \textit{define} the inverse-transformed fit as $\widecheck{a} := \Gamma_{-\varepsilon}\hat{a}\left(t;\Gamma_{\varepsilon}d\right)$. Similarly, this is the the ``transformed fit" in Figures \ref{fig:2} and \ref{fig:3}. Note the slight abuse of notation of using $\Gamma_\varepsilon$ to refer to the second component of $\Gamma_\varepsilon$ rather than the full two-variable function.
    \end{nota}
    \begin{defn}
        We define the \textit{disagreement coefficient} between models $A$ and $B$ as $\Delta_\mathrm{A,B}(c,\varepsilon) = \lVert \widecheck{a} - \widecheck{b} \rVert$ as the $L^2[t_1,t_k]$ distance between $\widecheck{a}$ and $\widecheck{b}$. We use the disagreement coefficient to determine locally optimal values of $(\mathbf{c}, \varepsilon)$ to be used for efficient differentiation between models by finding local maxima of the $\Delta(\mathbf{c},\varepsilon)$ surface. Here, $\mathbf{c}$ refers to the vector of parameters of the symmetries involved.
    \end{defn}

        Intuitively, the disagreement coefficient $\Delta_\mathrm{A,B}(c,0) = \lVert \hat{a} - \hat{b} \rVert$ will be small when $\varepsilon=0$ assuming models $A$ and $B$ are both good fits, but can be used to reveal values of $\mathbf{c}$ where the differentiation is more apparent.
        
        Pictorially, the definition and notation above can be summarised by the figure below.
    \begin{figure}[ht]
    \[\begin{tikzcd}[column sep=0.8em, row sep = 1.em]
    	&&& d \\
    	& {\Gamma_\varepsilon^{\mathrm{A}}d} &&&& {\Gamma_\varepsilon^{\mathrm{B}}d} \\
    	\\
    	{\hat{a}} &&&&&& {\hat{b}} \\
    	\\
    	& {\widecheck{a}} &&&& {\widecheck{b}} \\
    	&&& {\Delta_\mathrm{A,B}}
    	\arrow["{\Gamma_\varepsilon^{\mathrm{B}}}"', from=1-4, to=2-6]
    	\arrow["{\Gamma_\varepsilon^{\mathrm{A}}}", from=1-4, to=2-2]
    	\arrow["{\text{fit b}}"', from=2-6, to=4-7]
    	\arrow["{\text{fit a}}", from=2-2, to=4-1]
    	\arrow["{\Gamma_{-\varepsilon}^{\mathrm{B}}}"', from=4-7, to=6-6]
    	\arrow["{\Gamma_{-\varepsilon}^{\mathrm{A}}}", from=4-1, to=6-2]
    	\arrow[dashed, from=6-6, to=7-4]
    	\arrow[dashed, from=6-2, to=7-4]
    \end{tikzcd}\]
    \caption{Pictorial Summary of Constructing the Disagreement Coefficient}
    \label{fig:4}
    \end{figure}
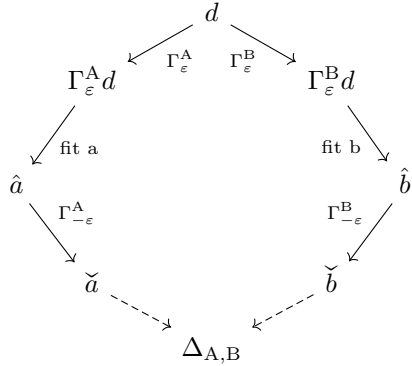

    \begin{rem}
        We note a few remarks about the disagreement coefficient. First, practical use of $\Delta$ is only valid away from discontinuities in the Lie symmetries that are used. For example, when computing the disagreement coefficient $\Delta_\mathrm{L,G}$ using the symmetries $\Pi_{\varepsilon}^{(\mathrm{L})}(t,y;c)$ and $\Pi_{\varepsilon}^{(\mathrm{G})}(t,y;c)$, we observe that $\Pi_{\varepsilon}^{(\mathrm{L})}(t,y;c)$ has a discontinuity when $ \left[\frac{1}{y} + \frac{c}{r_{\mathrm{L}}}e^{-r_{\mathrm{L}}t}\left(e^{-r_{\mathrm{L}}\varepsilon} - 1\right)\right] = 0$. The locus of such $(c,\varepsilon)$ is therefore a curve which must be avoided when carrying out a local search for optimal parameters.

        Second, we remark that the disagreement coefficient is calculated \textit{after} fitting model parameters. This suggests that simultaneously finding optimal model parameters and optimal symmetry parameters is a harder problem, but not one we attempt to answer.

        Third, we can naturally generalise the disagreement coefficient to accommodate models whose symmetries have different $\varepsilon$ and $\mathbf{c}$, although we will not do this in the example discussed.
    \end{rem}

    \paragraph{}
    We will consider an example of comparing the Gompertz model and the Richards model. We generate data by adding noise to the Richards model as before, and use $\Delta_\mathrm{G,R}$ to determine $c$ for which our framework produces clearer differentiation between the Gompertz and the Richards models.

    The $\Delta_\mathrm{G,R}$ plot in Figure \ref{fig:5a}, confirms that $\varepsilon = 0$ does not give any significant differentiation between the models, and neither does $c = 0$. However, what the figure does suggest is that increasing $c$ and $\varepsilon$ (in magnitude) results in better differentiation. Figure \ref{fig:5b} therefore shows the effect of the transformation at $\varepsilon = 0.1, c=10$.

    \begin{figure}[ht]
        \begin{subfigure}{0.48\textwidth}
            \includegraphics[width=\linewidth]{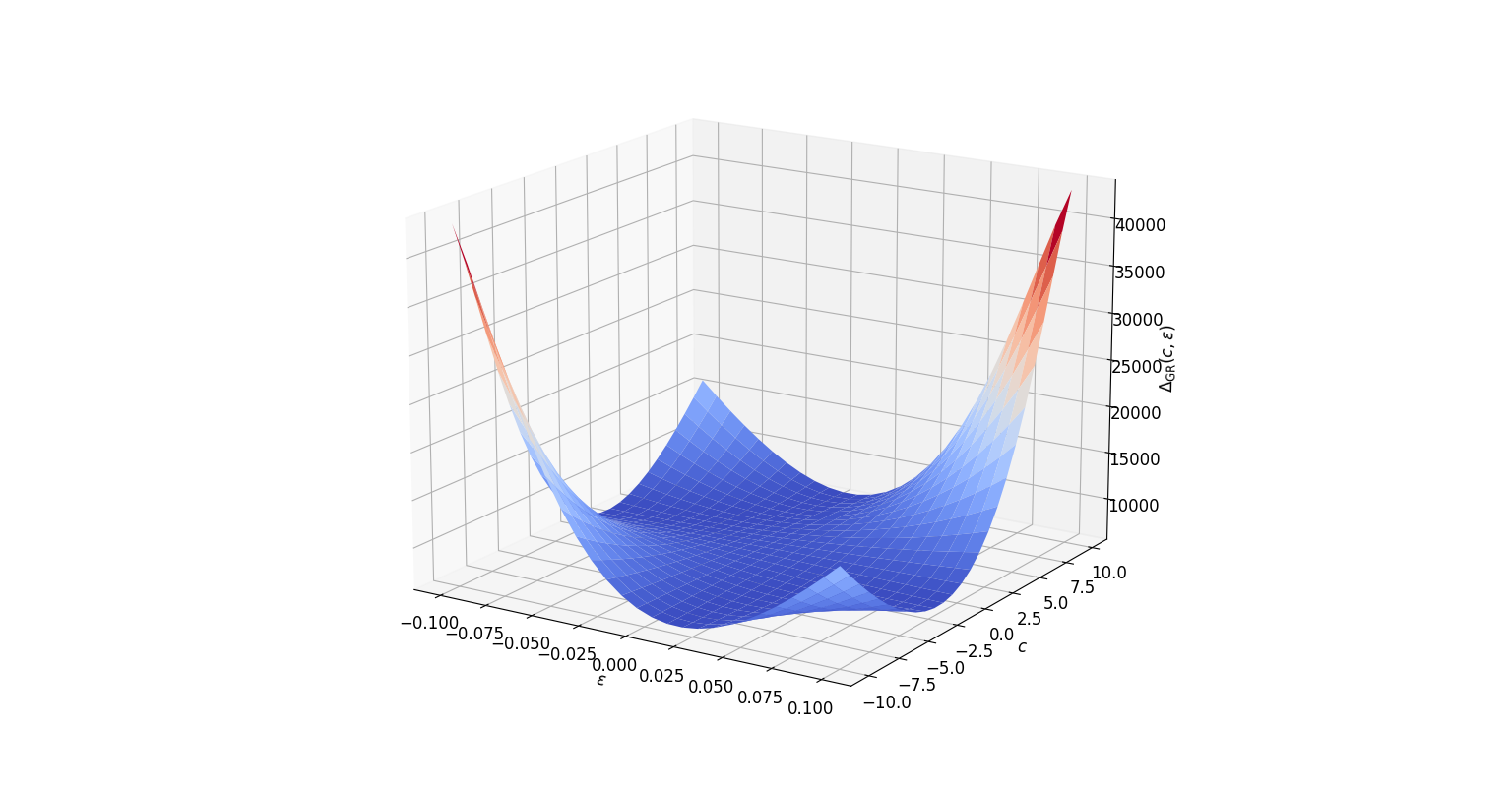}
            \caption{$\Delta_\mathrm{G,R}(c,\varepsilon)$} \label{fig:5a}
        \end{subfigure}
        \begin{subfigure}{0.48\textwidth}
            \includegraphics[width=\linewidth]{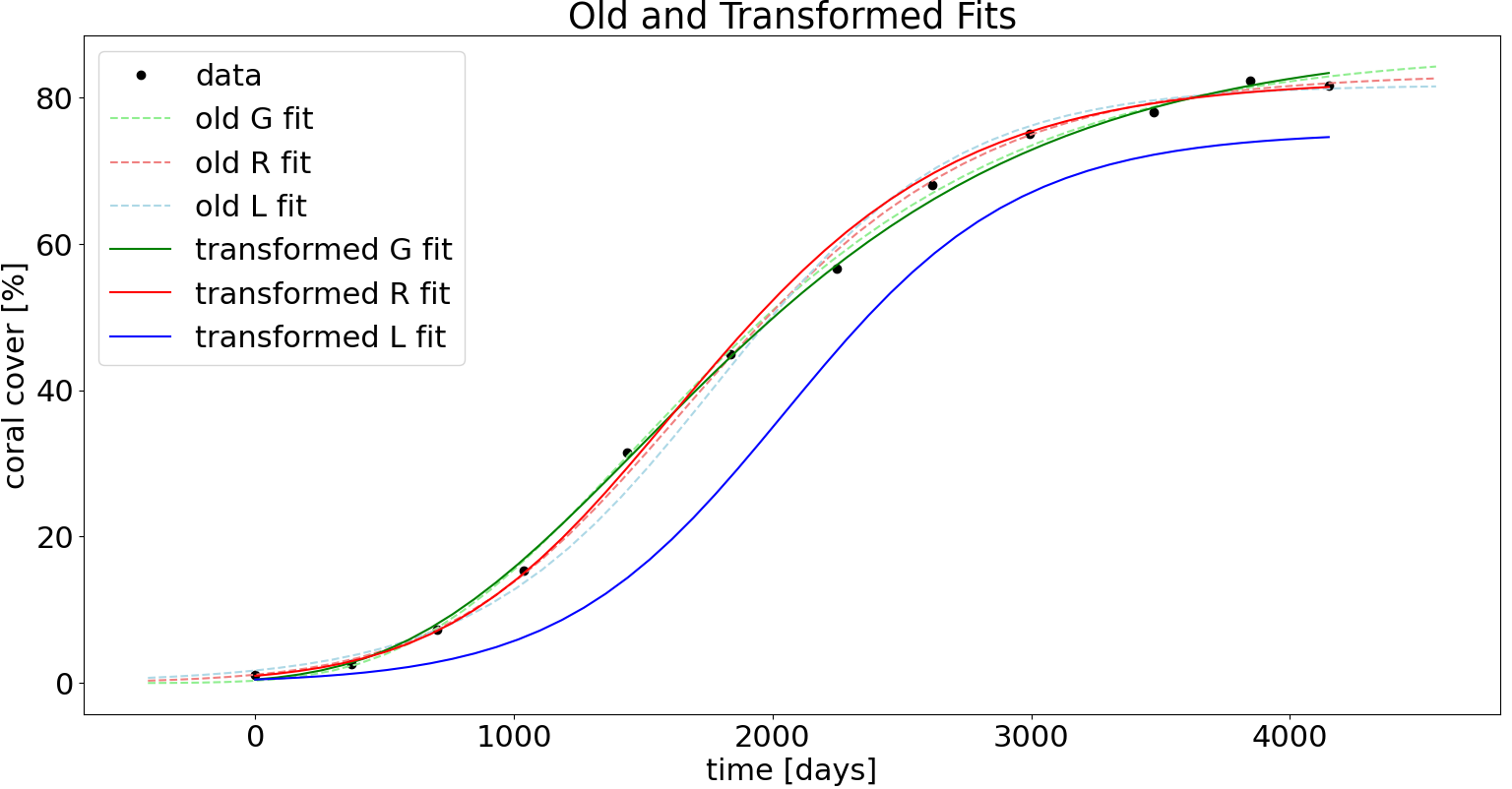}
            \caption{Simulated Richards Data} \label{fig:5b}
        \end{subfigure}\hspace*{\fill}
        \caption{Using the Disagreement Coefficient to Pick Parameters}
        \label{fig:5}
    \end{figure}

    \paragraph{}
    Recalling Definition \ref{symmetry definition}, all our analysis requires a neighbourhood  of $\varepsilon = 0$, $\mathcal{U} = (\varepsilon_{\mathrm{min}},\varepsilon_{\mathrm{max}})$, where the Lie symmetries we employ are analytic. However, there is no known practical method of determining these extremal values of $\varepsilon$. In our example (and in any example with an absence of a clear local maximum where the analysis is valid), the analysis above simply suggests $\varepsilon$ should be set as close to $ \varepsilon_{\mathrm{max}}$ (or $\varepsilon_{\mathrm{min}}$) as possible, while ensuring that the Lie symmetries used are continuous for all intermediate $s \in [0,\varepsilon]$ (respectively, $[\varepsilon,0]$), i.e. we choose
    \[
        \varepsilon = \argmax_{s \in [\varepsilon_{\mathrm{min}},\varepsilon_{\mathrm{max}}] \cap \mathcal{C}} \Delta_\mathrm{G,R}(\mathbf{c},s),
    \]
    defining $\mathcal{C} = \mathcal{C^+} \cup \mathcal{C^-}$, where
    \begin{align*}
        \mathcal{C^+} &= \left\{\varepsilon>0 \colon \Pi_{s}^{(\mathrm{G})} \text{ and } \Pi_{s}^{(\mathrm{R})} \text{ are continuous } \forall s \in [0,\varepsilon] \right\}, \text{ and} \\
        \mathcal{C^-} &= \left\{\varepsilon<0 \colon \Pi_{s}^{(\mathrm{G})} \text{ and } \Pi_{s}^{(\mathrm{R})} \text{ are continuous } \forall s \in [\varepsilon,0] \right\}.
    \end{align*}
    
    We therefore pose the following question.
    \begin{ques}
        \label{Q1}
        Is there a theoretical way to determine $\varepsilon_\mathrm{min}$ and $\varepsilon_\mathrm{max}$? If not, can bounds for these be found computationally?
    \end{ques}

\section{Discussion}

\paragraph{}
To summarise, the paper tackles two open problems in the field. The first being the common obstruction of \textit{finding symmetries} for a given first order ODE model, as the set of determining equations is typically underdetermined. The application to population growth models demonstrates the capabilities of the trivial symmetry in our framework for model selection. Existence trivial symmetries for models formulated with PDEs or systems of ODEs suggests the ability to potentially avoid needing to find model-specific Lie symmetries beyond the trivial ones, even in more complicated models. In the setting where finding symmetries is possible, the second problem that is tackled is that of \textit{choosing parameters} when employing multi-parameter symmetries for model selection, which is done by numerical analysis of the disagreement coefficient. 

Before concluding, we suggest some avenues of further research.
\begin{enumerate}
    \item First, there is some analysis to be done to ensure the validity of using independent Gaussian noise when simulating the data, for example, by applying the methodology in \cite{Lam} to test for residual autocorrelation in the transformed data.
    \item Second, other fitting techniques could be implemented, especially to study model selection with incomplete datasets \cite{Hai}. 
    \item \label{unified-Richards} Third, the model selection can be carried out with a broader family of sigmoid curves, such as the \textit{unified-Richards model} from \cite{Tjø}, which generalises the models considered in this paper.
    \item Finally, other new symmetries of the models of interest can be found (e.g. by combining \ref{unified-Richards}. and Theorem \ref{LMST}) to allow application of the techniques from subsection \ref{Delta subsection}.
\end{enumerate}

\paragraph{}
In addition to Question \ref{Q1}, we conclude the paper with two other open questions. The first is inspired by the results of \cite{FFR}. It is hoped that by answering this question, we can determine, given noisy data, the suitability of the symmetry framework for model selection, since we expect our framework to fail at selecting the correct underlying model when the data is too noisy. 

\begin{ques}
    \label{Q2}
    How can bounds for the low-noise regime (c.f. \ref{low noise}) be theoretically or practically determined?
\end{ques}

Before posing the final question, we draw attention to \cite{Hyd} Chapter 7 where symmetries of systems of ODEs, and Chapter 8 where symmetries of PDEs are discussed. For the full details of the theory, \cite{Olv} Chapter 2 introduces most, if not all, of the framework necessary to generalise the techniques from this paper to systems of PDEs. Some work has also been done on developing the theory of symmetries for SDEs \cite{Una, Koz1, Koz2}.

By answering the following question, it is hoped that the powerful machinery of Lie symmetries could be used to demonstrate the potential contribution of the symmetry-based framework to model selection in a broader context.

\begin{ques}
    Can the trivial symmetry be used to investigate the model selection problem in more complex settings, such as in models formulated with PDEs? What about models formulated with systems of differential equations? Do these ideas extend naturally to models formulated with SDEs?
\end{ques}

\begin{ack}
    The author wishes to thank Johannes Borgqvist for his excellent guidance and support, as well as Alex Browning and Anubhab Ghosal for many helpful discussions. This work was supported by funding from Hertford College and the Crankstart Scholarship.
\end{ack}

\bibliography{biblio}

\newcommand{\etalchar}[1]{$^{#1}$}
\begin{thebibliography}{FFRD{\etalchar{+}}23}

\bibitem[Aka98]{Aka}
H.~Akaike.
\newblock {\em Information Theory and an Extension of the Maximum Likelihood Principle}, pages 199--213.
\newblock Springer, New York, NY, 1998.

\bibitem[Aug21]{Aug}
F.~Augustsson.
\newblock Symmetries of mathematical models in biology.
\newblock Master's thesis, Chalmers University of Technology, Gothenburg, Sweden, 2021.

\bibitem[BAH11]{Bur}
K.~P. Burnham, D.~R. Anderson, and K.~P. Huyvaert.
\newblock Aic model selection and multimodel inference in behavioral ecology: some background, observations, and comparisons.
\newblock {\em Behav. Ecol. Sociobiol.}, 65:23--35, August 2011.

\bibitem[Boz87]{Boz}
H.~Bozdogan.
\newblock Model selection and {A}kaike's {I}nformation {C}riterion ({AIC}): The general theory and its analytical extensions.
\newblock {\em Psychometrika}, 52:345--370, September 1987.

\bibitem[BP22]{Bor}
J.~G. Borgqvist and S.~Palmer.
\newblock Occam’s razor gets a new edge: the use of symmetries in model selection.
\newblock {\em J. R. Soc. Interface.}, 19(193), August 2022.

\bibitem[CdB20]{CdB}
M.~Castro and R.~J. de~Boer.
\newblock Testing structural identifiability by a simple scaling method.
\newblock {\em PLoS Comput. Biol.}, 16(11), November 2020.

\bibitem[cha18]{MSE}
\url{https://math.stackexchange.com/users/168595} chaffdog.
\newblock Multivariable uniform convergence and differentiation.
\newblock Mathematics Stack Exchange \url{https://math.stackexchange.com/q/2513107}, March 2018.
\newblock (Version: 2018-03-04. Accessed: 2023-09-26).

\bibitem[CTK03]{CTK}
E.~S. Cheb-Terrab and T.~Kolokolnikov.
\newblock First-order ordinary differential equations, symmetries and linear transformations.
\newblock {\em Eur. J. Appl. Math.}, 14(2), May 2003.

\bibitem[FFRD{\etalchar{+}}23]{FFR}
O.~Fajardo-Fontiveros, I.~Reichardt, H.~R. {De Los Ríos}, J.~Duch, M.~{Sales-Pardo}, and R.~Guimerà.
\newblock Fundamental limits to learning closed-form mathematical models from data.
\newblock {\em Nat. Commun.}, 14(1043), February 2023.

\bibitem[Ger13]{Ger}
P.~Gerlee.
\newblock The model muddle: in search of tumor growth laws.
\newblock {\em Cancer Res.}, 73(8):2407--2411, April 2013.

\bibitem[Hai68]{Hai}
Y.~Haitovsky.
\newblock Missing data in regression analysis.
\newblock {\em J. R. Stat. Soc., B: Stat. Methodol.}, 30(1), 1968.

\bibitem[Hyd00]{Hyd}
P.~E. Hydon.
\newblock {\em Symmetry Methods for Differential Equations: A Beginner's Guide}.
\newblock Cambridge University Press, Cambridge, United Kingdom, 2000.

\bibitem[IDCC03]{Ive}
A.~R. Ives, B.~Dennis, K.~Cottingham, and S.~Carpenter.
\newblock Estimating community stability and ecological interactions from time-series data.
\newblock {\em Ecol. Monogr.}, 73, May 2003.

\bibitem[Koz12]{Koz1}
R.~Kozlov.
\newblock On symmetries of stochastic differential equations.
\newblock {\em Commun. Nonlinear Sci. Numer. Simul.}, 17(12), December 2012.

\bibitem[Koz18]{Koz2}
R.~Kozlov.
\newblock Random lie symmetries of itô stochastic differential equations.
\newblock {\em J. Phys. A: Math.}, 51(30), June 2018.

\bibitem[Lai64]{Lai}
A.~Kane Laird.
\newblock Dynamics of tumour growth.
\newblock {\em Br. J. Cancer}, 18(3), June 1964.

\bibitem[Lie80]{Lie1}
S.~Lie.
\newblock Theorie der {T}ransformationsgruppen {I}.
\newblock {\em Mathematische Annalen}, 16:441--528, 1880.

\bibitem[Lie90]{Lie2}
S.~Lie.
\newblock {\em Theorie der {T}ransformationsgruppen {II}}.
\newblock Teubner, Leipzig, Germany, 1890.

\bibitem[Lie93]{Lie3}
S.~Lie.
\newblock {\em Theorie der {T}ransformationsgruppen {III}}.
\newblock Teubner, Leipzig, Germany, 1893.

\bibitem[LLR{\etalchar{+}}23]{Lam}
B.~Lambert, C.~Lok Lei, M.~Robinson, M.~Clerx, R.~Creswell, S.~Ghosh, S.~Tavener, and D.~J. Gavaghan.
\newblock Autocorrelated measurement processes and inference for ordinary differential equation models of biological systems.
\newblock {\em J. R. Soc. Interface.}, 20(199), February 2023.

\bibitem[MTK15]{Mer}
B.~Merkt, J.~Timmer, and D.~Kaschek.
\newblock Higher-order lie symmetries in identifiability and predictability analysis of dynamic models.
\newblock {\em Phys. Rev. E.}, 92(1), July 2015.

\bibitem[MV20]{MV}
G.~Massonis and A.~F. Villaverde.
\newblock Finding and breaking lie symmetries: Implications for structural identifiability and observability in biological modelling.
\newblock {\em Symmetry}, 12(3), March 2020.

\bibitem[O{\etalchar{+}}17]{Osb}
K.~Osborne et~al.
\newblock Delayed coral recovery in a warming ocean.
\newblock {\em Glob. Chang. Biol.}, 23(9), September 2017.

\bibitem[OBC20]{OBC}
F.~Ohlsson, J.~G. Borgqvist, and M.~Cvijovic.
\newblock Symmetry structures in dynamic models of biochemical systems.
\newblock {\em J. R. Soc. Interface.}, 17(168), July 2020.

\bibitem[Olv86]{Olv}
P.~J. Olver.
\newblock {\em Applications of Lie Groups to Differential Equations}.
\newblock Springer, New York, NY, 1986.

\bibitem[Rud76]{Rud}
W.~Rudin.
\newblock {\em Principles of Mathematical Analysis}.
\newblock McGraw-Hill, New York, NY, 3rd edition, 1976.

\bibitem[S{\etalchar{+}}22]{Sim}
M.~J. Simpson et~al.
\newblock Parameter identifiability and model selection for sigmoid population growth models.
\newblock {\em J. Theor. Biol.}, 535, February 2022.

\bibitem[Sch78]{Sch}
G.~Schwarz.
\newblock Estimating the dimension of a model.
\newblock {\em Ann. Stat.}, 6(2), March 1978.

\bibitem[TML20]{Tho}
A.~Thompson, K.~Martin, and M.~Logan.
\newblock Development of the coral index, a summary of coral reef resilience as a guide for management.
\newblock {\em J. Environ. Manage.}, 271, October 2020.

\bibitem[TT10]{Tjø}
E.~Tjørve and K.~M.C. Tjørve.
\newblock A unified approach to the {R}ichards-model family for use in growth analyses: Why we need only two model forms.
\newblock {\em J. Theor. Biol.}, 267(3), December 2010.

\bibitem[{\"{U}}na03]{Una}
G.~{\"{U}}nal.
\newblock Symmetries of itô and stratonovich dynamical systems and their conserved quantities.
\newblock {\em Nonlinear Dyn.}, 32, June 2003.

\bibitem[W{\etalchar{+}}21]{War}
D.~J. Warne et~al.
\newblock Identification of two-phase recovery for interpretation of coral reef monitoring data.
\newblock {\em J. Appl. Ecol.}, 59(1), October 2021.

\bibitem[YEC09]{Yat}
J.~W.~T. Yates, N.~D. Evans, and M.~J. Chappell.
\newblock Structural identifiability analysis via symmetries of differential equations.
\newblock {\em Automatica}, 45(11), November 2009.

\end{thebibliography}
\bibliographystyle{alpha}

\appendix
\section{Appendix: Finding Other Symmetries}
\label{appendix}

    \paragraph{}
    Another method of finding symmetries (specifically for the Gompertz model, or other \textit{limiting} models) is via the following theorem, which gives sufficient conditions under which a symmetry of a model can be obtained by taking limits of a more general model.
    
    \begin{defn}
        We say that a model $\frac{d^ky}{dt^k} = \Omega(t,y,\frac{dy}{dt},\dots,\frac{d^{k-1}y}{dt^{k-1}};\boldsymbol{\Theta})$ with parameters $\boldsymbol{\Theta} = \left(\Theta_1,\dots,\Theta_n\right)$ is a \textit{generalisation} of $\frac{d^ky}{dt^k} = \omega(t,y,\frac{dy}{dt},\dots,\frac{d^{k-1}y}{dt^{k-1}};\boldsymbol{\theta})$ with parameters $\boldsymbol{\theta} = \left(\theta_1,\dots,\theta_m\right)$ if a particular choice, $\Tilde{\boldsymbol{\Theta}}$, of $\boldsymbol{\Theta}$ gives
        \[
            \omega(t,y,\frac{dy}{dt},\dots,\frac{d^{k-1}y}{dt^{k-1}};\boldsymbol{\theta}) = \Omega(t,y,\frac{dy}{dt},\dots,\frac{d^{k-1}y}{dt^{k-1}};\Tilde{\boldsymbol{\Theta}}).
        \]
    \end{defn}
    
    \begin{exam}
        The Richards model, $\dot{y} = \Omega(t,y;\boldsymbol{\Theta}) = r_{\mathrm{R}}y\left(1 - \left(\frac{y}{K}\right)^{\beta}\right)$, (with $\beta$ free) is a generalisation of the logistic model, $\dot{y} = \omega_{\mathrm{L}}(t,y;\boldsymbol{\theta}_{\mathrm{L}}) = r_{\mathrm{L}}y\left(1 - \frac{y}{K}\right)$, by setting 
        \[    
            \boldsymbol{\Theta} = (K,y_0,r_{\mathrm{R}},\beta) = (K,y_0,r_{\mathrm{L}},1).
        \]
        It also a generalisation of the Gompertz model, $\dot{y} = \omega_{\mathrm{G}}(t,y;\boldsymbol{\theta}_{\mathrm{G}}) = r_{\mathrm{G}}y\operatorname{log}\left(\frac{K}{y}\right)$ in the limit
        \[
            r_{\mathrm{G}}y\operatorname{log}\left(\frac{K}{y}\right) = \lim_{\substack{\beta r_{\mathrm{R}} \to r_{\mathrm{G}}\\ \beta \to 0}} r_{\mathrm{R}}y\left(1 - \left(\frac{y}{K}\right)^{\beta}\right).
        \]
        In this case, we call the Gompertz model a \textit{limiting model} of the Richards model.
    \end{exam}
    \begin{thm}
        \label{LMST}
        Suppose $\frac{d^ky}{dt^k} = \Omega(t,y,\frac{dy}{dt},\dots,\frac{d^{k-1}y}{dt^{k-1}};\boldsymbol{\Theta})$ is a generalisation of $\frac{d^ky}{dt^k} = \omega(t,y,\frac{dy}{dt},\dots,\frac{d^{k-1}y}{dt^{k-1}};\boldsymbol{\theta})$ by a choice of $\boldsymbol{\Theta}$ involving a limit. Further suppose that the $k$th jet space, $\mathcal{J}^{(k)}$, is convex and that, under this limit, the convergence of $\nabla\Omega$ is uniform on $\mathcal{J}^{(k)}$, then any symmetry of $\Omega$ descends to a symmetry of $\omega$ in the limit.
    \end{thm}
    \begin{proof} (For ODEs, although it is possible to generalise to PDEs, systems of ODEs, etc.)

        First denote $\frac{d^ny}{dt^n}$ by $y^{(n)}$, then the linearised symmetry condition for infinitesimals $\Xi$ and $H$ of the model $y^{(k)} = \Omega(t,y,y^{(1)},\dots,y^{(k-1)};\boldsymbol{\Theta})$ reads
        \[
            H^{(k)} = \Xi \Omega_t + H \Omega_y + H^{(1)} \Omega_{y^{(1)}} + \dots + H^{(k-1)} \Omega_{y^{(k-1)}}. \tag{LSC$\Omega$}
        \]
        Similarly, the linearised symmetry condition for the infinitesimals $\xi$ and $\eta$ of the model $y^{(k)} = \omega(t,y,y^{(1)},\dots,y^{(k-1)};\boldsymbol{\theta})$ is
        \[
            \eta^{(k)} = \xi \omega_t + \eta \omega_y + \eta^{(1)} \omega_{y^{(1)}} + \dots + \eta^{(k-1)} \omega_{y^{(k-1)}} \tag{LSC$\omega$}.
        \]
        Since $\Omega$ generalises $\omega$, we know that $\Omega \to \omega$ converges pointwise in the limit. By assumption, it now follows that $\Omega \to \omega$ converges \textit{uniformly} and that $\nabla\Omega \to \nabla\omega$ by a generalisation of Theorem 7.17 in \cite{Rud} due to \cite{MSE}. Hence, any solution $(\Xi,H)$ to (LSC$\Omega$) will tend to a solution of (LSC$\omega$) in the limit.

    \end{proof}

\end{document}